\documentclass[11pt]{article}

\usepackage{geometry}
\usepackage{amsmath,amsfonts,amssymb,bbm}
\usepackage{tikz}
\usetikzlibrary{arrows}
\usepackage{float}
\usepackage[linesnumbered, boxed, ruled]{algorithm2e}
\usepackage{hyperref}
\hypersetup{
	colorlinks=true,
	linkcolor=[rgb]{0,0,0.4},
	citecolor=[rgb]{0, 0.4, 0},
	urlcolor=[rgb]{0.6, 0, 0}
}
\usepackage{amsthm} 

\newtheorem{theorem}{Theorem}[section]
\newtheorem{corollary}[theorem]{Corollary}
\newtheorem{lemma}[theorem]{Lemma}
\newtheorem{observation}[theorem]{Observation}

\newtheorem{proposition}[theorem]{Proposition}

\newtheorem{definition}{Definition}

\newcommand{\MyFrame}[1]{\noindent \framebox[\textwidth]{ \begin{minipage}{0.97\textwidth} #1 \end{minipage}}}%
\let\oldnl\nl
\newcommand{\nonl}{\renewcommand{\nl}{\let\nl\oldnl}}

\newcommand{\cardinality}[1]{\left\vert{#1}\right\vert}
\newcommand{\expect}[2]{\mathop{\mathbb{E}}_{#1} \left[ #2 \right]}

\newcommand{\reals}{\mathbb{R}}

\DeclareMathOperator*{\argmax}{arg\,max}
\newcommand{\ind}{\mathbbm{1}}

\newcommand{\items}{M}
\newcommand{\val}{v}
\newcommand{\hrep}{w}

\newcommand{\M}{\mathcal{M}}

\newcommand{\dist}{D}
\newcommand{\distH}{D'}

\newcommand{\feas}{\mathcal{C}}
\newcommand{\feasSet}{\mathcal{S}}

\newcommand{\cutoff}{t}

\newcommand{\rev}[1]{#1R{\small EV}}

\newcommand{\BREV}{\mbox{\rev{B}}}
\newcommand{\SREV}{\mbox{\rev{S}}}

\newcommand{\REV}{\mbox{\rev{}}}

\newcommand{\OPTcopies}{\mbox{OPT}^{\mbox{copies}}}

\newcommand{\SINGLE}{\mbox{SINGLE}}
\newcommand{\CORE}{\mbox{CORE}}
\newcommand{\TAIL}{\mbox{TAIL}}
\newcommand{\NONFAV}{\mbox{NON-FAVORITE}}

\newcommand{\virtVal}{\varphi}

\begin{document}

\title{A Simple and Approximately Optimal Mechanism\\ for a Buyer with Complements}

\author{Alon Eden\thanks{Computer Science, Tel-Aviv University. \texttt{alonarden@gmail.com}. Work done in part while the author was visiting the Simons Institute for the Theory of Computing.} \and Michal Feldman\thanks{Computer Science, Tel-Aviv University, and Microsoft Research. \texttt{michal.feldman@cs.tau.ac.il}.} \and Ophir Friedler\thanks{Computer Science, Tel-Aviv University. \texttt{ophirfriedler@gmail.com}.} \and Inbal Talgam-Cohen\thanks{Computer Science, Hebrew University of Jerusalem. \texttt{inbaltalgam@gmail.com}.} \and S.~Matthew Weinberg\thanks{Computer Science, Princeton University. \texttt{smweinberg@princeton.edu}. Work done in part while the author was a research fellow at the Simons Institute for the Theory of Computing.}}

\maketitle

\thispagestyle{empty}\maketitle\setcounter{page}{0}

\begin{abstract}
We consider a revenue-maximizing seller with $m$ heterogeneous items and a single buyer whose valuation $v$ for the items may exhibit both substitutes (i.e., for some $S, T$, $v(S \cup T) < v(S) + v(T)$) and complements (i.e., for some $S, T$, $v(S \cup T) > v(S) + v(T)$). 
We show that the mechanism first proposed by Babaioff et al.~[2014] - 
the better of selling the items separately and bundling them together - 
guarantees a $\Theta(d)$ fraction of the optimal revenue, where $d$ is a measure on the degree of complementarity. Note that this is the first approximately optimal mechanism for a buyer whose valuation exhibits any kind of complementarity, and extends the work of Rubinstein and Weinberg [2015], which proved that the same simple mechanisms achieve a constant factor approximation when buyer valuations are subadditive, the most general class of complement-free valuations.

Our proof is enabled by the recent duality framework developed in Cai et al.~[2016], which we use to obtain a bound on the optimal revenue in this setting. Our main technical contributions are specialized to handle the intricacies of settings with complements, and include an algorithm for partitioning edges in a hypergraph. Even nailing down the right model and notion of ``degree of complementarity'' to obtain meaningful results is of interest, as the natural extensions of previous definitions provably fail.
\end{abstract}

\newpage

\section{Introduction} \label{sec:intro}

Consider a revenue-maximizing seller with $m$ items to sell to a single buyer. When there is just a single item, and the buyer's value is drawn from some distribution with CDF $F$, seminal works of Myerson~\cite{myerson1981optimal}, and Riley and Zeckhauser~\cite{riley1983optimal} prove that the optimal mechanism is to simply set whatever price maximizes $p \cdot (1-F(p))$. It soon became well-understood that beyond the single-item setting, the optimal mechanism suffers many undesireable properties that make it unusable in practice, including randomization, non-monotonicity, and others~\cite{RochetC98,hart2013menu,hart2012maximal,briest2010pricing,daskalakis2013mechanism,daskalakis2014complexity,thanassoulis2004haggling,Pavlov11a}.
Following seminal work of Chawla, Hartline, and Kleinberg~\cite{chawla2007algorithmic}, there is now a sizeable body of research proving that the simple mechanisms we see in practice are in fact approximately optimal in quite general settings, helping to explain their widespread use~\cite{chawla2010multi,chawla2010power,kleinberg2012matroid,hart2012approximate,babaioff2014simple,bateni2015revenue,rubinstein2015simple,li2013revenue,yao2015n,CaiZ16,chawla2016mechanism}.

Still, prior work has largely been limited to additive\footnote{A buyer valuation is additive if $v(S) = \sum_{i \in S} v(\{i\})$.} or unit-demand\footnote{A buyer valuation is unit-demand if $v(S) = \max_{i \in S} \{v(\{i\})\}$.} buyers. Only recently have researchers begun tackling more complex valuation functions, and even these works have remained restricted to subclasses of \emph{subadditive} valuations, also called \emph{complement-free}~\cite{rubinstein2015simple,chawla2016mechanism,CaiZ16}.\footnote{A valuation is subadditive if $v(S \cup T) \leq v(S) + v(T)$ for all $S, T$.}
While subadditive valuations are quite general, they can only capture interaction between items as \emph{substitutes}. For example, if the items are pieces of furniture, a buyer's marginal valuation for a chair might decrease as her home gets more and more filled due to lack of space. To date, no results in this line of work model iteraction between items as \emph{complements}. For example, a buyer's value for a kitchen table might actually increase if she already has a chair to sit. The goal of this paper is to study simple and approximately optimal mechanisms in domains (like the example above) where buyer valuations exhibit both substitutes and complements.

\paragraph{Buyers with Complements.} Even for the traditionally simpler domain of \emph{welfare} maximization, the state-of-the-art only recently has begun designing mechanisms for buyers with complements~\cite{abraham2012combinatorial,feige2015unifying,feldman2015combinatorial,feldman2016simple}.
The main difficulty is that horrible lower bounds are known for general valuations~\cite{nisan2006communication}, so in order to get interesting positive results, some assumptions are necessary on the degree to which buyer valuations exhibit substitutes or complements.
Interestingly, good positive results are possible in the complete absence of complements and no restriction on the degree of substitutability~\cite{dobzinski2005approximation,dobzinski2007two,feige2009maximizing,feldman2013simultaneous,devanur2015simple}, but not vice versa: many strong lower bounds still exist in the absence of substitutes but with arbitrary complementarity~\cite{lehmann2002truth,abraham2012combinatorial,morgenstern2015market,feldman2016simple}.

So the goal of these recent works is to parameterize the ``degree of complementarity'' that a valuation function admits, and prove an approximation guarantee of $f(d)$ whenever buyer valuations have ``complementarity of degree at most $d$''~\cite{abraham2012combinatorial,feige2013welfare,feige2015unifying,feldman2015combinatorial,feldman2016simple}.
For example, if you were selling a table, chair, bicycle, banana, and socks, you would reasonably expect buyers to view the table and chair as complements, but likely not the bicycle and banana. Similarly, you wouldn't expect any set of three items to be viewed as complements (outside of what's already captured by the table and chair as a pair). So it seems overly pessimistic not to try and exploit this. Ideally, a good formal definition for ``complements of degree $d$'' should make sense in its own right (i.e. without appealing to results) and capture a smooth transition as $d$ grows (i.e. we don't have $f(0) = 1$ and $f(d) = m$ for all $d > 0$).  Interestingly, the right formal definition of ``complements of degree $d$'' seems to differ between environments. Some examples of previous successful definitions include the ``supermodular degree'' and ``positive-hypergraphs degree''~\cite{feige2013welfare,feige2015unifying}.

Nailing down the right model of complementarity degree is even trickier for the revenue objective, as we must also fold some notion of independence into the value distribution in order to avoid extremely strong lower bounds that hold against even additive valuations over two items~\cite{briest2010pricing,hart2013menu}.\footnote{Specifically, there exists a distribution $\dist$ over $\reals^2$ such that when a single additive buyer's valuation is drawn from $\dist$, the optimal revenue for the seller is infinite, but the revenue of the best deterministic mechanism is $1$.} We postpone a formal definition of our model and corresponding notion of complementarity degree to Section~\ref{sec:prelim}, and give an illustrative example here. Imagine you are selling furniture and related goods. Some items naturally exhibit complementarities: with a table, chair, and silverware, a buyer can eat meals at home. With a table, four chairs, and a game of Settlers, they could host a board games night. With two sofas and a TV, they could instead host a movie night. So think of the buyer as having a non-negative valuation for being able to eat meals in their home, host events, etc., and their preliminary value for a set $S$ of items is additive over the activities that $S$ allows them to partake in.\footnote{Such valuations functions are called ``positive-hypergraph'' (PH) valuations.} But there's a catch: no buyer has room in their apartment to comfortably fit a table, TV, two sofas, and four chairs. So the items are also substitutes - even if the buyer were to wind up with the entire warehouse of furniture, they won't get use out of anything besides what fits in their apartment. So we let $\feas$ represent a set system determining which items can fit in the apartment, and let $\hrep(T)$ denote the buyer's value for whatever special activity the items in exactly $T$ allow her to partake in (that she couldn't partake in with any proper subset).
The buyer's value for a set of items $S$ is the maximum over all $S'$ that fit in her apartment of the sum of her values for all the activities she can partake in using $S'$. Independence enters the model by assuming buyers have independent values for different activities, and the degree of complementarity is captured via the maximum number of activities that require any given item.

\paragraph{Main Result.}

Our main result (Theorem~\ref{thm:main}) is that the mechanism proposed by Babaioff et al.~\cite{babaioff2014simple} - 
the better of selling separately (post a price on each item, let the buyer purchase whatever subset she likes) or bundling together (post a single price on the grand bundle, let the buyer purchase or not) - achieves a tight $\Theta(d)$ approximation whenever buyer valuations exhibit complementarity at most $d$. 

We also complete the picture by showing that our notion of complementarity is in some sense the right one: if instead we measure complementarity via the ``supermodular degree,'' then there exist populations in our model with supermodular degree $d$ for which the better of selling separately and bundling together achieves only a $\Omega(2^d/d)$-approximation. Similarly, if we instead measure complementarity via the ``positive-hypergraph degree,'' then there exist populations in our model with positive-hypergraph degree $d$ for which the better of selling separately and bundling together achieves only a $\Omega(\sum_{\ell \leq d}\binom{m}{\ell}/m)$-approximation. Both notions of degree are defined formally in Section~\ref{sec:lowerBound} where the lower bounds are proved. The point is not that $\Theta(d)$ is a ``better'' bound than $\Omega(2^d/d)$, as this is in some sense not a fair comparison, but rather that ``supermodular degree'' and ``positive-hypergraph degree'' are incapable of capturing the smooth transition from low degrees to high degrees of complementarities as they can only take on $m$ different values but provide guarantees that range from $1$ to $\Omega(2^m)$. In comparison, our notion of degree of complementarity takes on $2^{m-1}$ different values, and provides guarantees that range from $1$ to $\Omega(2^m)$, allowing for an exponentially finer-grained tradeoff.
\paragraph{Our Techniques.}
Our starting point is a duality-based upper bound on the optimal achievable revenue coming from recent work of~\cite{cai2016duality}. Their upper bound decomposes into three parts, which they call \textsc{SINGLE}, \textsc{CORE}, and \textsc{TAIL}. So the goal is to show that selling separately well-approximates \textsc{SINGLE}, and that bundling together well-approximates \text{CORE} and \textsc{TAIL}. Fortunately, the analysis of~\cite{cai2016duality} is fairly robust, and we are able to prove that bundling together achieves a constant factor of both \textsc{CORE} and \textsc{TAIL} via a similar approach.
Our main technical contribution appears in Section~\ref{sec:single}, where we prove that selling separately gets an $O(d)$-approximation to \textsc{SINGLE}.
Incidentally, bounding \textsc{SINGLE} happened to be the easiest part of the analysis in~\cite{cai2016duality} for additive valuations.

Without getting into details about what exactly this \textsc{SINGLE} term is, we can still highlight the key challenge. Essentially, we would like to post a different price on each activity. In fact, we can show that the optimal ``activity-pricing scheme'' even obtains a constant-factor approximation to \textsc{SINGLE}.
	The catch is that we sell items, not activities. We may wish to set drastically different prices on many different activities requiring the same item, and it's unclear that we can achieve the desired activity prices by cleverly setting prices on the items separately (in fact, it could be impossible).
	So our main technical contribution is an algorithm to find a subset of activities $S$ for which it is possible to achieve any desired activity-pricing on $S$ by only posting prices on items, and the optimal revenue from activities in $S$ is a $d$-approximation to the optimal {\em activity}-pricing scheme. It turns out that the right sets of activities to search for are ones where each activity requires an item \emph{not} required by any of the others, that the number of collections with this property necessary to partition all activities tightly characterizes the approximation guarantee of selling separately, and that $d$ collections suffice whenever each item is required by at most $d$ activities.
\subsection{Related Work}
\paragraph{Multi-Dimensional Auction Design.}
A rapidly growing body of recent literature has shown that simple mechanisms are approximately optimal in quite general settings~\cite{chawla2007algorithmic,chawla2010multi,chawla2010power,kleinberg2012matroid,hart2012approximate,li2013revenue,babaioff2014simple,yao2015n,rubinstein2015simple,bateni2015revenue,chawla2016mechanism}. Of these, the result most related to ours is~\cite{rubinstein2015simple}, which proves that the better of selling separately and bundling together achieves a constant-factor approximation for a single buyer whose valuation is drawn from a population that is ``subadditive with independent items''. Their model is similar to our model with $d=1$ (but neither subsumes the other), so our results can best be interpreted as an extension of theirs to buyers whose valuations also exhibit complementarity. 

In terms of techniques, our work makes use of a recent duality framework developed in~\cite{cai2016duality}. The same duality framework has been used in concurrent work by the present authors to prove multi-dimensional ``Bulow-Klemperer'' results~\cite{eden2016competition}, and independent work by others to design simple, approximately optimal auctions for multiple subadditive bidders~\cite{CaiZ16}. Still, the duality theory is only used to provide an upper bound on the revenue in all these cases, and the remaining technical contributions are disjoint. In particular, for the present paper, Section~\ref{subsec:virtValAndBenchmark} has a high technical overlap with these works, and Section~\ref{sec:non-favorite} bears some similarity. But our main technical contribution lies in Section~\ref{sec:single}, which is unique to the problem at hand.
{\paragraph{Agents with Complements.}
In recent years there has also been a rapid growth in the design of algorithms and mechanisms in the presence of complements  \cite{abraham2012combinatorial,feige2013welfare,feldman2014constrained,feldman2015building,feige2015unifying,feldman2015combinatorial,feldman2016simple}.
These works consider many different aspects: for example, assuming strategic behavior of agents (or not), assuming the existence of strict substitutes (or not), or focusing on simple mechanisms and quantifying the efficiency of equilibria. 
In all these works, some notion of degree of complementarity was cast on a class of valuation functions, and the approximation ratio guaranteed grew as a function of complementarity degree. 
It is noteworthy that quite often different settings motivate different degrees of complementarity to best capture the degradation in possible guarantees. For instance, \cite{abraham2012combinatorial} uses the positive hypergraph (PH) degree, \cite{feige2013welfare} uses the supermodular degree, \cite{feige2015unifying,feldman2015combinatorial} use the maximum over PH degree, and \cite{feldman2016simple} uses the positive supermodular degree.

In comparison to this literature, ours is the first to consider \emph{revenue} maximization for buyers with complements.

\subsection{Discussion and Future Work}
We present the first simple and approximately optimal mechanism for a buyer whose valuation exhibits both substitutes and complements. We show that for a natural notion of ``degree of complementarity,'' the better of selling separately and selling together achieves a tight $\Theta(d)$-approximation to the optimal revenue. We provide rigorous evidence that this is ``the right'' notion to consider via large lower bounds for classes of valuations that previous definitions would deem simple.

Our main technical contribution is an algorithm to partition a collection of sets into subcollections such that each set (in the subcollection) contains an item not contained in the others (in that same subcollection). Due to the robustness of previously-developed tools like the ``core-tail'' decomposition~\cite{li2013revenue,babaioff2014simple,rubinstein2015simple,yao2015n,chawla2016mechanism}, and duality-based benchmarks~\cite{cai2016duality}, we are able to focus our technical contributions to the specific problem at hand.

The obvious direction for future work would be to see whether simple mechanisms remain approximately optimal for \emph{multiple} buyers with complementarity degree $d$. Doing so would likely require at least one substantial innovation beyond the ideas in this paper, as even the $d=1$ case remains open (even considering the recent breakthrough result of~\cite{CaiZ16}). Our work also contributes to the growing body of evidence that our community now has the tools to ``catch up'' the state-of-the-art for multi-dimensional mechanism design to the wealth of knowledge that currently exists for single-dimensional settings. Considering buyers with complements is one important path in this direction, but there are numerous others as well.

\section{Preliminaries} \label{sec:prelim}
\paragraph{Buyer Valuations.} We consider a setting in which a seller wishes to sell a set $\items$ of $m$ items to a single buyer.
The buyer has a valuation function $\val$ that assigns a non-negative real number $\val(S)$ to every bundle of items $S\subseteq \items$.
The valuation is normalized ($\val(\emptyset)=0$) and monotone ($\val(S) \leq \val(T)$ whenever $S \subseteq T$). We also abuse notation and let $\val(X) = \expect{S \leftarrow X}{\val(S)}$ when $X$ is a random set.

\paragraph{Complementarities.} An increasingly popular model to represent complementarities is via a \emph{positive hypergraph representation}. That is, $\hrep: 2^\items \to \reals^+$ is a non-negative function, and $\hrep(T)$ denotes the bonus valuation that the consumer enjoys from exactly the set of items $T$ (in addition to the value the consumer already enjoys for proper subsets of $T$), i.e., $\val(S) = \sum_{T \subseteq S} \hrep(T)$. In the language of Section~\ref{sec:intro}, $\hrep(T)$ denotes the bidder's value for the activity requiring exactly items in $T$. We will sometimes refer to $T$ as a \emph{hyperedge}, thinking of $\hrep(\cdot)$ as weight function on the hypergraph with nodes $\items$. As an example, if $\val$ is additive, then defining $\hrep(\{i\}) = \val(\{i\})$ and $\hrep(T) = 0$ whenever $|T| > 1$ yields $\val(S) = \sum_{T \subseteq S} \hrep(T)$. We say that $\val$ (or $\hrep$) exhibits \emph{complementarities of degree $d$} if for all $i$, $\cardinality{\{S \ni i : \hrep(S) > 0\}} \leq d$.

\paragraph{Substitutes.} An equally popular model to represent substitutes is via combinatorial constraints. Let $\feas \subseteq 2^\items$ denote a downwards-closed set system on $\items$. $S \notin \feas$ denotes that at least some items in $S$ are substitutes, and the buyer does not derive value from all of $S$. Many valuations that exhibit \emph{only} substitutabilities are ``additive subject to constraints $\feas$'': $\val(S) = \max_{T \subseteq S, T \in \feas}\{\sum_{i \in T} \val(\{i\})\}$. For example, unit-demand valuations can be represented with $\feas = \{T : \cardinality{T} \leq 1\}$.

\paragraph{Complements and Substitutes.} We choose to model substitutes and complementarities together by combining the above two models. That is, there is a positive hypergraph representation $\hrep$ that represents complementarities, and combinatorial constraints $\feas$ that represent substitutabilities, and $\val(S) = \max_{T \subseteq S, T \in \feas} \{\sum_{U \subseteq T} \hrep(U)\}$. Recall the furniture example: $\hrep$ represents that two sofas and a TV allow you to host a movie night, whereas any proper subset doesn't. $\feas$ represents that you can only fit so much furniture in your apartment. We assume w.l.o.g. that $\hrep(T) = 0$ for all $T \notin \mathcal{C}$, as the bidder will never be able to partake in activity $T$ no matter what (because the required items don't fit in the apartment).

\paragraph{Value Distributions.} We model our buyer valuation $\val(\cdot)$ as being drawn from the population $\dist$ in the following way. There are some constraints $\feas$, that are fixed (not randomly drawn). Each $\hrep(T)$ is then drawn independently from some distribution $\distH_T$ for all $T$, and $\val(S) = \max_{T \subseteq S, T \in \feas}\{\sum_{U \subseteq T} \hrep(U)\}$. We say that $\dist$ has complementarity $d$ if all $\val$ in the support of $\dist$ have complementarity $d$. Note that this implies $\dist$ has complementarity $d$ if and only if for all $i$, $\cardinality{\{T \ni i : \Pr[\hrep(T)=0] < 1\}} \leq d$. We use $V$ to denote the support of $\dist$, $f(\val)$ to denote $\Pr_{\hat{\val}\leftarrow \dist}[\hat{\val} = \val]$, and $f_T(y) = \Pr_{x \gets \distH_T}[y = x] $. 

\paragraph{Truthful Mechanisms and Revenue Maximization.} Formally, a mechanism $\M$ has two mappings $X: V \rightarrow \Delta(2^\items)$, and $p:V \rightarrow \reals$. $X$ takes as input a valuation $\val$ and awards a (potentially random) subset of items. $p$ takes as input a valuation $\val$ and charges a price. $\M$ is then \emph{truthful} if for all $\val, \val ' \in V$, $\val(X(\val)) - p(\val) \geq \val(X(\val ')) - p(\val ')$.\footnote{Note that for a single buyer, there is no need to distinguish between Bayesian Incentive Compatible and Dominant Strategy Incentive Compatible - they're the same.} Alternatively, one can view a mechanism as a \emph{menu} that lists options of the form $(X, p)$, where $X \in \Delta(2^\items)$ and $p \in \reals$. A buyer with value $\val(\cdot)$ then selects the menu option $\argmax\{\val(X) - p\}$. It is easy to see the equivalence between the two representations: simply setting $(X(\val),p(\val)) = \argmax\{\val(X) - p\}$ takes one from the menu view to a truthful mechanism. We denote by $\REV (\dist)$ the optimal revenue attainable by any truthful mechanism when buyer valuations are drawn from the population $\dist$.

\paragraph{Simple Mechanisms.} The two simple mechanisms we study are selling separately (\SREV) and bundling together (\BREV). We denote by $\BREV(\dist)$ the optimal expected revenue attainable by selling all items together, and will drop the parameter $\dist$ when it is clear from context. It is well-known that $\BREV(\dist) = \max p\cdot \Pr\left[\val(\items) \geq p\right]$~\cite{myerson1981optimal}. \SREV{} is a touch trickier, as it is NP-hard for buyers in our model to even decide what set of items they wish to purchase at a given set of prices, so it's not even clear how we should evaluate the ``revenue'' of a price vector. We cope with this using a similar approach to~\cite{rubinstein2015simple}: we define $\SREV^*$ to be the optimal revenue attainable by any item pricing \emph{only counting an item as sold if every set the buyer is willing to purchase contains that item}. More formally, for a given item pricing $\vec{p}$, and valuation $\val$, let $P_i(\vec{p}, \val) = 1$ if $\exists S \ni i, \val(S) - \sum_{j \in S} p_j > 0$ and $\forall S \not \ni i, \val(S) - \sum_{j \in S}p_j \leq 0$, and $P_i(\vec{p},\val) = 0$ otherwise. Then $\SREV^*(\dist) = \max_{\vec{p}} \expect{\val\leftarrow \dist}{\sum_i P_i(\vec{p},\val)\cdot p_i}$.

\paragraph{Discrete vs. Continuous Distributions.} Like~\cite{cai2016duality}, we only explicitly consider distributions with finite support. Like their results, all of our results immediately extend to continuous distributions as well via a discretization argument of~\cite{daskalakis2012symmetries,rubinstein2015simple,hartline2010bayesian,hartline2011bayesian,bei2011bayesian}. We refer the reader to~\cite{cai2016duality} for the formal statement and proof. Theorem~\ref{thm:pricing} assumes that for every single-dimensional random variable $X$ and number $q\in [0, 1]$, there exists a threshold $p$ so that $X\geq p$ with probability {\em exactly} $q$, which might a priori seem problematic for discrete distributions. Fortunately, standard ``smoothing'' techniques allow this assumption to be valid for discrete distributions. A formal discussion of this appears in Remark 2.4 of~\cite{rubinstein2015simple}.

\paragraph{The Copies Environment.} In our bounds, we'll make use of a related ``copies environment''~\cite{chawla2007algorithmic,chawla2010multi,chawla2010power,kleinberg2012matroid}. For any product distribution $\distH = \times_{i=1}^k \distH_i$, we define the corresponding copies setting as follows: there is a single item for sale, and $k$ buyers. Buyer $i$'s value for the item is drawn from the distribution $\distH_i$.
{For instance, in our model, the hypergraph representation of the valuation is drawn from $\distH = \times_{S}{\distH_S}$, so we would have a bidder for every subset, with bidder $S$'s value drawn from the distribution $\distH_S$.}

We can then define the benchmark $\OPTcopies(\distH)$ to be the expected revenue obtained by the optimal mechanism (Myerson's~\cite{myerson1981optimal}) on input $\distH$.
Note that this is equal to $\expect{\hrep \leftarrow \distH}{\max_T \{\bar{\virtVal}_T(\hrep(T)),0\}}$, where $\bar{\virtVal}_T(\cdot)$ denotes Myerson's \emph{ironed virtual value} for the distribution $\distH_T$. We will make use of the following theorem from~\cite{chawla2010multi}:

\begin{theorem}[\cite{chawla2010multi}]\label{thm:pricing}
For any $q \leq 1$, there exist (possibly random) prices $\{p_T\}_T$ such that:
    \begin{enumerate}
        \item Revenue is high:
        $\OPTcopies(\distH) \leq \frac{1}{q}\sum_{T\subseteq \items} \expect{p_T}{p_T \cdot \Pr_{x \leftarrow \distH_T}[x \geq p_T]}$.
        \label{thm:pricing:upperBound}
        \item Probability of sale is low:        $\sum_{T\subseteq \items}{\expect{p_T}{\Pr_{x \leftarrow \distH_T}[x \geq p_T]}} \leq q$

        \label{thm:pricing:probUpperBound}
        \item Moreover, each $p_T$ takes on at most two values. If $\distH_T$ is regular, then $p_T$ is a point-mass.\footnote{We will not actually make use of bullet 3 other than to simplify notation, but it might help remind some readers where these prices come from.}
        \label{thm:pricing:twoPrices}
    \end{enumerate}
\end{theorem}

\section{Our Duality Benchmark and Main Theorem Statement} \label{subsec:virtValAndBenchmark}
We extend the duality framework of~\cite{cai2016duality} to our setting in a natural manner. 
Full technical details are deferred to Appendix~\ref{app:duality}. 
The only technical detail needed for stating our revenue benchmark is the following:
we partition the valuation space $V$ into $2^m-1$ different regions, depending on which hyperedge is the most valuable to a buyer with valuation $\val$.
Specifically, we say that $\val$ is in region $R_A$ if $A = \argmax_{T \subseteq \items} \{\hrep(T) \}$, with ties broken lexicographically.
\begin{corollary}\label{cor:benchmark}
{For valuation distribution $\dist$ established by drawing a hypergraph representation $\hrep \gets \prod_S \distH_S$ and returning $\val(S) = \max_{T \subseteq S, T \in \feas}\{\sum_{U \subseteq T} \hrep(U)\}$. 
}
\begin{align*}
\REV(\dist) \leq &
{\expect{\val \leftarrow \dist}{\max_{S \in \mathcal{C}}\{\sum_{T \subseteq S} \hrep(T)\cdot\ind\left[\val \notin R_T \right]\}}}&(\NONFAV)
\\
& +\expect{\val \gets \dist}{\sum_{S \subseteq \items}
\max\{0,\bar{\virtVal}_{S}(\hrep(S))\}\cdot \ind\left[\val \in R_S
\right]} & (\SINGLE)
\end{align*}
\end{corollary}

In Section~\ref{sec:single}, we show that $\max\{\SREV^*,\BREV\}$ gets a $4(d+1)$-approximation to $\SINGLE$. This portion of the analysis develops techniques specific to buyers with restricted complements. In Section~\ref{sec:non-favorite}, {we show that $\BREV$ gets a $12$-approximation to \NONFAV}. This portion of the analysis will look somewhat standard to the reader familiar with~\cite{cai2016duality}, with a little extra work to extend their main ideas to our setting. We conclude this section with our main theorem, whose proof will be completed by the end of Section~\ref{sec:non-favorite}:

\begin{theorem} \label{thm:main}
{	For a distribution $\dist$ that has complementary $d$, 
$\REV \leq (4d + 16)\max\{\BREV, \SREV^*\}.$
}
\end{theorem}

\section{Bounding \SINGLE}\label{sec:single}

In this section, we prove that the better of selling separately and selling together gets an $O(d)$ approximation to \SINGLE.
\begin{proposition}\label{prop:single}
$\SINGLE \leq 4d \SREV^* + 4\BREV .$
\end{proposition}

We begin by relating $\SINGLE$ to $\OPTcopies$:

\begin{observation}\label{obs:copies}
$\SINGLE \leq \OPTcopies .$
\end{observation}
\begin{proof}
    First, observe that there is exactly one $S$ for which $\ind[\val \in R_S] = 1$. So it is certainly the case that for all $\val$ (with $\val(S) = \sum_{T \subseteq S} \hrep(T)$), we have:

    $$\sum_{S \subseteq \items} \max\{0, \bar{\virtVal}_S(\hrep(S))\} \cdot \ind [ \val \in R_S] \leq \max_{S \subseteq \items}\{0,\bar{\virtVal}_S(\hrep(S))\}.$$
    $$\Rightarrow \mathbb{E}_{\val \leftarrow \dist}\left[ \sum_{S \subseteq \items} \max \{0, \bar{\virtVal}_S(\hrep(S)) \} \cdot \ind [\val \in R_S] \right]\leq \mathbb{E}_{\val \leftarrow \dist} \left[ \max_{S \subseteq \items} \{0, \bar{\virtVal}_S(\hrep(S))\}\right].$$

    Above, the LHS is exactly $\SINGLE$, and the RHS is exactly $\OPTcopies$.
\end{proof}

Note that if the buyer's valuation were additive, at this point we'd already be finished. We could simply set the prices guaranteed by Theorem~\ref{thm:pricing} and be done. As we consider more complex buyer valuations, there are two barriers we must overcome. The first is due to substitutability: if we try to set prices on each subset separately, just because the buyer is \emph{willing} to purchase set $S$ doesn't mean he will \emph{choose} to purchase set $S$, because he may purchase some substitutes instead. Note that this issue doesn't arise in absense of substitutes: if the buyer is willing to purchase $S$ by itself, he is certainly willing to add $S$ to any other set of purchased items. The second barrier is due to complementarity: even once we decide the ``correct'' price to charge for set $S$, we can only set prices on \emph{items} and not on \emph{bundles}. Therefore, the prices we want to set for different bundles necessarily interfere with each other. This is the novel barrier unique to values with complementarity, and is also the only part of the analysis where the (necessary) factor of $d$ arises. 

The first step to overcoming the complements barrier is to find a subset of bundles for which we can still set the appropriate prices. As a warm-up, let's see what the argument would look like assuming that there were only complements and no substitutes ($\feas = 2^\items$):

\begin{lemma}\label{lem:okay}
Let $\feas = 2^\items$ and $T_1,\ldots, T_k$ be subsets of $\items$ such that $T_i \not \subseteq \cup_{j \neq i} T_j$ for all $i$. Then for all $\{p_T\}_{T \subseteq \items}$, $\SREV \geq \sum_i p_{T_i}\Pr_{x \leftarrow \distH_{T_i}}[x \geq p_{T_i}]$. 
\end{lemma}
\begin{proof}
Set price $p_{T_i}$ on the item contained in $T_i$ but not $\cup_{j \neq i} T_j$ (if there are multiple, select one arbitrarily). Then by hypothesis, the price the bidder would have to pay in order to receive the entire set $T_i$ is exactly $p_{T_i}$. Because $\feas = 2^\items$, whenever $\hrep(T_i) \geq p_{T_i}$, the buyer will choose to purchase the set $T_i$ in addition to whatever else they choose to purchase. Therefore, the item contained in $T_i$ but not $\cup_{j \neq i} T_j$ is purchased with probability at least $\Pr_{x \leftarrow \dist_{T_i}}[x \geq p_{T_i}]$, and the revenue of this item pricing is at least $\sum_i p_{T_i} \Pr_{x \leftarrow \distH_{T_i}}[x \geq p_{T_i}]$. 
\end{proof}
 
The proof of Lemma~\ref{lem:okay} makes use of the assumption that $\feas = 2^\items$ in exactly one place: to argue that whenever $\hrep(T_i) \geq p_{T_i}$, the buyer chooses to purchase the complete set $T_i$. When $\feas \neq 2^\items$, it may be the case that even though the buyer is willing to purchase set $T_i$, she chooses to purchase substitutes instead. We can remove this assumption on $\feas$ by restricting attention to certain price vectors.

\begin{lemma}\label{lem:stillokay}
Let $\feas$ be any downwards closed set system and $T_1,\ldots, T_k$ be subsets of $\items$ such that $T_i \not \subseteq \cup_{j \neq i} T_j$ for all $i$. Then for all $\{p_T\}_{T \subseteq \items}$ such that $p_T \geq 4 \BREV$ for all $T$, $\SREV^* \geq \tfrac{1}{4} \sum_i p_{T_i}\Pr_{x \leftarrow \distH_{T_i}}[x \geq p_{T_i}]$. 
\end{lemma}
\begin{proof}
	Set price $p_{T_i}/2$ on the item contained in $T_i$ but not $\cup_{j \neq i} T_j$ (if there are multiple, again select one arbitrarily).
	The price the bidder would have to pay in order to receive the entire set $T_i$ is exactly $p_{T_i}/2$.
	Suppose $\hrep(T_i) \geq p_{T_i}$. Then, the buyer is not only willing to purchase $T_i$, but also gets utility at least $p_{T_i}/2$ for doing so.
	The only reason she would choose not to purchase this set is if there were some other set $S$ with $T_i \not \subseteq S$ and $\val(S) \geq p_{T_i}/2 \geq 2\BREV$. As $\val(S) \leq \val(\items)-\hrep(T_i)$ for all such $S$, in order for such a set to exist, it must be the case that $\val(\items) - \hrep(T_i) \geq 2 \BREV$. Clearly, this occurs with probability at most $\frac{1}{2}$, as otherwise we could set price $2\BREV$ on the grand bundle, sell with probability strictly larger than $\frac{1}{2}$ and make revenue strictly larger than $\BREV$. Moreover, $\val(\items) - \hrep(T_i) = \sum_{U \neq T_i} \hrep(U)$ is completely independent of $\hrep(T_i)$. Therefore, even conditioned on $\hrep(T_i) \geq p_{T_i}$, the probability that the bidder is interested in some other set $S$ with $T_i \not \subseteq S$ is at most $\frac{1}{2}$, and therefore the buyer indeed chooses to purchase $T_i$ with probability at least $\Pr_{x \leftarrow \distH_{T_i}}[x \geq p_{T_i}]\cdot \frac{1}{2}$.
\end{proof}

Finally, we can combine Lemma~\ref{lem:stillokay} with Theorem~\ref{thm:pricing} to reduce our search to the problem of partitioning the hyperedges into collections $H_x = \{T_{x1},\ldots, T_{xk_x}\}$ such that $T_{xi} \not \subseteq \cup_{j \neq i} T_{xj}$ for all $i$.

\begin{corollary}\label{cor:stillokay}
Let $\feas$ be any downwards closed set system, and let $\{H_x\}_{x \in [k]}$ be a partition of the hyperedges $\{T : f_T(0) < 1\}$ such that for all $x$, and all $T \in H_x$, $T \not \subseteq \cup_{T' \in H_x \setminus \{T\}} T'$. Then $4k\SREV^* + 4\BREV \geq \SINGLE$.
\end{corollary}
\begin{proof}
	Take $q=1$ in Theorem~\ref{thm:pricing} and let $\{p_T\}_{T \subseteq \items}$ be the guaranteed (randomized) prices. By Theorem~\ref{thm:pricing} condition~\ref{thm:pricing:twoPrices}, there exist two deterministic prices $p^H_T \geq p^L_T$ and probabilities $q_T$ such that $p_T = p^H_T$ with probability $q_T$, and $p_T = p^L_T$ with probability $1-q_T$. Therefore, Theorem~\ref{thm:pricing} condition~\ref{thm:pricing:upperBound} can be rewritten as:
	$$\OPTcopies \leq \sum_{T \subseteq \items} q_T p^H_T \cdot \Pr_{x \leftarrow \distH_T}[x \geq p^H_T] + (1-q_T) p^L_T \cdot \Pr_{x \leftarrow \distH_T}[x \geq p^L_T]$$
	
	We can further rewrite this by breaking up the two sums into prices that exceed $4\BREV$, and those that don't,
	let $\mathcal{B} = 4\BREV$ for simplicity:
	
	\begin{align*}
	\OPTcopies \leq& \sum_{T \subseteq \items, p^H_T \leq \mathcal{B}} q_T p^H_T \cdot \Pr_{x \leftarrow \distH_T}[x \geq p^H_T] +  \sum_{T \subseteq \items, p^L_T \leq \mathcal{B}} (1-q_T) p^L_T \cdot \Pr_{x \leftarrow \distH_T}[x \geq p^L_T]\\
	&+ \sum_{T \subseteq \items, p^H_T > \mathcal{B}} q_T p^H_T \cdot \Pr_{x \leftarrow \distH_T}[x \geq p^H_T] +  \sum_{T \subseteq \items, p^L_T > \mathcal{B}} (1-q_T) p^L_T \cdot \Pr_{x \leftarrow \distH_T}[x \geq p^L_T]
	\end{align*}
	By condition~\ref{thm:pricing:probUpperBound} of Theorem~\ref{thm:pricing}, we have
	$$\sum_{T \subseteq \items}q_T\cdot \Pr_{x \leftarrow \distH_T}[x \geq p^H_T] + (1-q_T)\cdot \Pr_{x \leftarrow \distH_T}[x \geq p^L_T] \leq 1$$
	Therefore, as all prices in the top sum above are at most $\mathcal{B}$, the entire top two terms sum to at most $\mathcal{B} = 4 \BREV$.
	
	For the bottom two terms, there is no term for $T$ if $p^H_T \leq \mathcal{B}$. If $p^H_T > \mathcal{B} \geq p^L_T$, define $p_T = p^H_T$.
	If $p^H_T > p^L_T > \mathcal{B}$, then set $p_T$ to whichever of $\{p^H_T,p^L_T\}$ maximizes $p_T\cdot \Pr_{x \leftarrow \distH_T}[x \geq p_T]$. Then $\sum_{T \subseteq \items, p^H_T > \mathcal{B}} p_T \cdot \Pr_{x \leftarrow \distH_T}[x \geq p_T]$ is at least as large as the bottom two terms above.
	Moreover, as all $p_T > \mathcal{B}$, we can apply Lemma~\ref{lem:stillokay} to conclude that for all $T_1,\ldots, T_k$ such that $T_i \not \subseteq \cup_{j \neq i} T_j$ for all $i$, $\SREV^* \geq 1/4 \sum_i p_{T_i} \Pr_{x \leftarrow \distH_{T_i}}[x \geq p_{T_i}]$.
	
	Finally, as $\{H_x\}_{x \in [k]}$ partitions the hyperedges so that for all $x$ and $T \in H_x$, $T \not \subseteq \cup_{T' \in H_x \setminus \{T\}} T'$, we get:
	$$\sum_{T \subseteq \items, p^H_T > \mathcal{B}} p_T \cdot \Pr_{x \leftarrow \distH_T}[x \geq p_T] = \sum_{x = 1}^k \sum_{T \in H_x, p^H_T > \mathcal{B}} p_T\cdot \Pr_{x \leftarrow \distH_T}[x \geq p_T] \leq 4k \cdot \SREV^*$$
	
	The last inequality is due to Lemma~\ref{lem:stillokay}, and completes the proof.
\end{proof}

So the last remaining task is to find a good partition of hyperedges, such that within each partition, every hyperedge contains at least one item not contained in the other hyperedges in the same partition. We isolate this contribution in Section~\ref{sec:partition} below.

\subsection{Partitioning Hyperedges with {Restricted} Complements}\label{sec:partition}
\begin{figure}[H]
	\MyFrame{
		$\mathrm{Partition\mbox{-}Edges}$\\
		\textbf{Input:} List of hyperedges, $E \subseteq 2^\items$.\\
		\textbf{Output:} A partition of $E$ into $\{H_x\}_x$ such that for all $x$ and all $T \in H_x$, $T \not \subseteq \cup_{T' \in H_x \setminus \{T\}} T'$.
		\begin{enumerate}
			\item $E_{\mathrm{curr}}\gets E$, $ i \gets 0$. \label{step:init}
			\item While $E_{\mathrm{curr}} \neq \emptyset$:\label{step:build_sets}
			\begin{enumerate}
				\item $i \gets i+1$
				\item $E_i\gets E_{\mathrm{curr}}$.
				\item For each $T\in E_i$ (in arbitrary order): If $T \subseteq \bigcup_{S\in E_i\setminus \{T\}}{S}$ Then  $E_i\gets E_i\setminus \{T\}$. \label{step:edgeViolation}
				\item $E_{\mathrm{curr}}\gets E_{\mathrm{curr}}\setminus E_i$.\label{step:remove}
			\end{enumerate}
			\item Return the partition $\{E_j\}_{j \in [i]}$.
		\end{enumerate}
	}
	\caption{An edge partitioning process.}
	\label{alg:edge_selection}
\end{figure}

{We provide a high-level description of our algorithm here, and give pseudocode in Figure~\ref{alg:edge_selection}. Recall that the algorithm takes as input a set of hyperedges, and returns a partition of the hyperedges $\{H_x\}_x$, so that in each partition $H_x$, every hyperedge $S \in H_x$ contains an item that is not in any other hyperedge $T \in H_x$. The algorithm iteratively constructs each $H_x$, and initially initializes $H_x$ to contain all remaining hyperedges. Then, it iteratively eliminates all ``bad'' hyperedges (those that \emph{don't} contain an item absent from the others) until the remaining hyperedges have the desired property. In the proof of Theorem~\ref{thm:edgeSelect} below, it is easy to show that the algorithm outputs a feasible partition, and the trick is guaranteeing that each iteration makes sufficient progress towards finalizing the partition.
}

\begin{theorem} \label{thm:edgeSelect}
For any set of hyperedges $E \subseteq 2^\items$, Algorithm~\ref{alg:edge_selection} returns a partition of $E = \{H_x\}_{x \in [k]}$ such that:
\begin{enumerate}
\item For all $x$, and all $T \in H_x$, $T \not \subseteq \cup_{T' \in H_x \setminus \{T\}} T'$.
\item $k \leq \max_{i}\{|\{T \in E : i \in T\}|\}$. 
\end{enumerate}
\end{theorem}
\begin{proof}
	First, it is clear that the algorithm indeed properly outputs a partition of $E$: observe that due to line~\ref{step:remove}, when a hyperedge is permanently assigned to some $E_i$, it will not be assigned to any $E_{i'}$, which implies that all the $E_i$'s are disjoint. Also, every hyperedge is either permanently assigned to some $E_i$, or remains in $E_{\mathrm{curr}}$, which, by line~\ref{step:build_sets} implies that the algorithm terminates only when every hyperedge is permanently assigned to some $E_i$. So every hyperedge is contained in some partition, and the partitions are disjoint.
	
	That the output partition satisfies Property 1) is easy to verify: For any $x$, $T \in H_x$ only the check in~\ref{step:edgeViolation} passes for $T$ and (the present) $H_x$. Once the check passes, some other edges will be removed from $H_x$ before the output. Clearly, removing edge from $H_x$ cannot cause $T$ to all of a sudden be contained in $\cup_{T' \in H_x \setminus \{T\}} T'$ when it was previously not contained. So Property 1) is satisfied.
	
	To prove Property 2), first denote by $E_{\textrm{curr}}^i$ the state of $E_{\textrm{curr}}$ at the start of iteration $i$. We will show that $\cup_{T \in E_i} T = \cup_{T \in E_{\textrm{curr}}^i} T$. In other words, every element contained in some hyperedge in $E_{\textrm{curr}}^i$ is still contained in some hyperedge in $E_i$. To see this, observe that when $E_i$ is first set to $E^i_{\textrm{curr}}$, we clearly have $\cup_{T \in E_i} T  = \cup_{T \in E^i_{\textrm{curr}}} T$. The only time hyperedges are removed from $E_i$ is in step~\ref{step:edgeViolation}. Note that in order for a hyperedge to be removed from $E_i$, it must be the case that $T \subseteq \cup_{T' \in E_i \setminus \{T\}} T'$. In other words, in order to remove $T$ from $E_i$, it must be that all the elements contained in $T$ are also contained in $\cup_{T' \in E_i \setminus \{T\}} T'$. Therefore, removing $T$ does not change $\cup_{T' \in E_i} T'$, and when we terminate, we maintain $\cup_{T \in E_i} T = \cup_{T \in E^i_{\textrm{curr}}}T$.
	
	To see why this implies Property 2), note that the above implies that if for any $i$, $|\{T \in E, i \in T\}| = d$, then $i$ will be contained in at least one hyperedge in all of $E_1, \ldots, E_d$, and therefore no hyperedges containing $i$ remain in $E_{\textrm{curr}}^{d+1}$. In particular, for $d = \max_{i}\{|\{T \in E, i \in T\}|\}$, it's the case that for all $i$, no hyperedges containing $i$ remain in $E_{\textrm{curr}}^{d+1}$, and therefore the algorithm terminates with at most $d$ partitions.
\end{proof}

We can now combine everything to provide a proof of Proposition~\ref{prop:single}:

\begin{proof}[Proof of Proposition~\ref{prop:single}]
Combining Theorem~\ref{thm:edgeSelect} with Corollary~\ref{cor:stillokay}, we get that whenever $\dist$ has complementarity $d$, that $4d\SREV^* + 4\BREV \geq \SINGLE$, completing the proof.
\end{proof}

\section{Bounding \NONFAV} \label{sec:non-favorite}
In this section, we bound $\NONFAV$ using similar ideas to those developed in~\cite{cai2016duality}. Much of the process will look familiar to experts famliar with~\cite{rubinstein2015simple, cai2016duality}, but there are a couple of new ideas sprinkled in. We begin by breaking $\NONFAV$ into $\CORE + \TAIL$, as is by now standard ($\cutoff$ will be chosen later). Omitted proofs appear in Appendix~\ref{app:proof:nonfav}.

\begin{lemma}\label{lem:coretail}
$\NONFAV$ is upper bounded by the following:	
\begin{align*} 
&\expect{\val \leftarrow \dist}{\max_{S \in \feas}\{\sum_{T \subseteq S} \hrep(T)\cdot \ind \left[\hrep(T) \leq \cutoff \right]\}} + && (\CORE) \nonumber \\
& \expect{\val \leftarrow \dist}{\sum_{S : \hrep(S) > \cutoff}\hrep(S)\cdot \ind\left[\val \notin R_S \right]} && (\TAIL)
\end{align*}
\end{lemma}

\paragraph{Bounding \CORE.}
Our main approach to bound $\CORE$ is to apply the same concentration bound of Schechtman~\cite{schechtman2003concentration} used in~\cite{rubinstein2015simple}. Essentially, we just have to show that our valuation functions are ``subadditive over independent items,'' for the appropriate definition of ``items'' (which happens to be hyperedges). It's perhaps not obvious that our valuation functions are subadditive over independent ``items,'' but indeed they are. 

Let's first recall the definition of subadditive over independent items. In the definition below, we intentionally write $N$ instead of $\items$ to denote the set of items, as the ``items'' in the definition may be different than the items for sale.

\begin{definition} \label{def:independentHyperedges} A distribution $\dist$ over valuation functions $\val : 2^{N} \rightarrow \reals$ is {\em subadditive over independent items} if the following conditions hold:
    \begin{enumerate}
        \item {\em No externalities and independence across items}:
        For every item $i$, let $\Omega_i$ be a compact subset of a normed space (i.e., $\Omega_i = [0,1]$). 
        There exists a product distribution $\dist'$ over $\times_{i \in N} \Omega_i$ (that is, $\dist' = \prod_{i \in N}\dist'_i$), and a collection of deterministic functions $V_{\feasSet} :\times_{i\in \feasSet}\Omega_i \rightarrow \reals$
        such that a sample $\val$ from $\dist$ can be drawn by sampling $\vec{x}\leftarrow \dist'$, and defining $\val(S) = V_S(\vec{x}_S)$.
        \item {\em Monotonicity}:
        Every $\val$ in the support of $\dist$ is monotone, i.e., $\val(\feasSet) \leq \val(\feasSet')$ for every $\feasSet \subseteq \feasSet'$.
        \item {\em Subadditivity}:
        Every $\val$ in the support of $\dist$ is subadditive, i.e., $\val(\feasSet \cup \feasSet') \leq \val(\feasSet') +\val(\feasSet')$, for all $S, S'$.
    \end{enumerate}
\end{definition}

\begin{definition} \label{def:lipschitz}
    Let $\dist$ denote a distribution over valuation functions, and $\dist'$ denote the product distribution and $\{V_S(\cdot)\}$ the deterministic functions that witness $\dist$ as subadditive over independent items. Then $\dist$ is $c$-{\em Lipschitz} if for all $\vec{x}, \vec{y}$, and sets of items $S, T$, we have:
    \begin{align*}
        \cardinality{V_S(\vec{x}_S) - V_T(\vec{y}_T)}\leq c\cdot \left(\cardinality{X \cup Y} - \cardinality{X \cap Y} + \cardinality{\{i \in X \cap Y : x_i \neq y_i  \} } \right)
    \end{align*}
\end{definition}
We use the following lemma and corollary (of a concentration inequality due to Schechtman~\cite{schechtman2003concentration}) from~\cite{rubinstein2015simple} (the bound in Corollary~\ref{cor:schechtman} is slightly improved from~\cite{rubinstein2015simple}, so we include a proof in Appendix~\ref{app:proof:nonfav}):
\begin{lemma}(\cite{rubinstein2015simple}) Let $\dist$ be a distribution that is subadditive over independent hyperedges, where for each hyperedge $T$, $\val(\{T\}) \in [0,c]$ with probability $1$. Then $\dist$ is $c$-Lipschitz.
\end{lemma}

\begin{corollary}\label{cor:schechtman}(\cite{rubinstein2015simple})
    Suppose that $\dist$ is a distribution that is subadditive over independent hyperedges and $c$-Lipschitz, if $a$ is the median of $\val(N)$, then
    $\expect{}{\val(N)} \leq 3a +c\cdot (2+1/\ln 2)$
\end{corollary}

Finally, we just need to relate $\CORE$ to a random variable that is subadditive over independent items.

\begin{lemma}\label{lem:subadditive}
$CORE$ is the expectation of a random variable $\val_{\CORE}(N)$, where $\val_{\CORE}(\cdot)$ is $\cutoff$-lipschitz and subadditive over independent items $N = 2^\items$. Moreover, $\val_{\CORE}(N)$ is stochastically dominated by $\val(\items)$. 
\end{lemma}
\begin{proof}
	Let the ``items'' $N = 2^\items$. Let the distributions $\hat{\dist}_T = \distH_T\cdot \ind[w(T) \leq \cutoff]$ (that is, a random variable drawn from $\hat{\dist}_T$ can be coupled with the random variable $w(T) \cdot \ind[w(T) \leq \cutoff]$). Define constraints $\feas' \subseteq 2^N$ ($=2^{2^\items}$) so that a subset $U$ of $2^\items$ is in $\feas'$ if and only if there exists a set $C \in \feas$ with $\cup_{T \in U} T \subseteq C$. In other words, $U \in \feas'$ if and only if the union of elements of $U$ is contained in some set in $\feas$. Finally, define $V_U(\vec{x}_U) = \max_{U' \subseteq U, U' \in \feas'}\{\sum_{T \in U} x_T\}$.
	
	It is easy to see that $\val_{\CORE}(\cdot)$ has no externalities and independent items. It is also easy to see that $\val_{\CORE}(\cdot)$ is monotone. Finally, we'll prove that $\val_{\CORE}(\cdot)$ is subadditive by observing that $\feas'$ is downwards-closed. To see this, simply observe that if $U' \subseteq U$, and $\cup_{T \in U} T \subseteq C$, then clearly $\cup_{T \in U'} T \subseteq C$. So if $C \in \feas$ witnesses that $U \in \feas'$ and $U' \subseteq U$, then $C$ also witnesses that $U' \in \feas'$.
	
	Now that $\feas'$ is downwards closed, it's easy to see (and well-known) that $\val_{\CORE}$ is subadditive: For any $U, W$, let $X = \argmax_{X' \subseteq U \cup W, X \in \feas'} \{\sum_{T \in X} x_T\}$. Then let $U' = X \cap U$, and $W' = X \cap W$. Clearly, $\sum_{T \in X} x_T \leq \sum_{T \in U'} x_T + \sum_{T \in W'} x_T$. As $\feas'$ is downwards closed, $U' \in \feas'$ and $W' \in \feas'$. Therefore, $\val_{\CORE}(W) + \val_{\CORE}(U) \geq \sum_{T \in U'} x_T + \sum_{T \in W'} x_T \geq \sum_{T \in X} x_T = \val_{\CORE}(U \cup W)$, and $\val_{\CORE}(\cdot)$ is subadditive.
	
	So finally, we just have to show that $\val_{\CORE}(N)$ is stochastically dominated by $\val(\items)$. Couple the random variable $x_T$ drawn from $\hat{\dist}_T$ so that $x_T = \hrep(T)\cdot \ind[\hrep(T) \leq \cutoff]$. Now consider $U^* = \argmax_{U \subseteq 2^\items, U \in \feas'}\{\sum_{T \in U} x_T\}$. Then we have $\val_{\CORE}(N) = \sum_{T \in U^*}x_T$. By definition of $\feas'$, there exists some $C \in \feas$ such that $T \subseteq C$ for all $T \in U^*$. Therefore:
	\begin{align*}
	\val_{\CORE}(N) = \sum_{T \in U^*} x_T &\leq \sum_{T \subseteq C} x_T & &\\
	&\leq \sum_{T \subseteq C} \hrep(T) & &\text{ (because $x_T \leq \hrep(T)$)}\\
	&\leq \max_{S \subseteq M, S \in \feas}\{\sum_{T \subseteq S} \hrep(T)\} & &\text{ (because $C \in \feas$)}\\
	& = \val(\items).
	\end{align*}
	
	So when $x_T$ and $\hrep(T)$ are coupled in this way, we have $\val_{\CORE}(N) \leq \val(\items)$, and therefore $\val(\items)$ stochastically dominates $\val_{\CORE}(N)$.
\end{proof}

Now, Lemma~\ref{lem:subadditive} combined with Corollary~\ref{cor:schechtman} essentially says that $3\cdot\val(\items)$ exceeds $\CORE - \cutoff\cdot(2+1/\ln 2)$ with probability at least $1/2$, allowing us to conclude with the following proposition:
\begin{proposition} \label{prop:core}
$       \CORE
        \leq
        6\BREV + \cutoff\cdot(2+1/\ln 2).
$
\end{proposition}

\begin{proof}
    Let $a$ be the median of the random variable $\val_{\CORE}(N)$. Then $\Pr[\val_{\CORE}(N) \geq a] = 1/2$. As $\val(\items)$ stochastically dominates $\val_{\CORE}(N)$, we have $\Pr[\val(\items) \geq a] \geq 1/2$. Moreover, by Corollary~\ref{cor:schechtman}, the fact that $\CORE = \mathbb{E}[\val_{\CORE}(N)]$, and that $\val_{\CORE}$ is $\cutoff$-lipschitz and subadditive over independent items, we have:

    $$\CORE \leq 3a + \cutoff(2+1/\ln 2).$$

    Moreover, as $\Pr[\val(\items) \geq a] \geq 1/2$, we have:

    $$\BREV \geq a/2.$$

    Combining the two above equations proves the proposition.
\end{proof}

\paragraph{Bounding \TAIL.}
Our approach to bound $\TAIL$ is again similar to~\cite{cai2016duality}. 
We begin by rewriting $\TAIL$ using linearity of expectation {and the fact that the hypergraph representation $\hrep$ of valuation $\val$ is drawn from $\distH$ which is a product distribution:}
\begin{align}
\TAIL
= &
\expect{\val \leftarrow \dist}{\sum_{T\subseteq \items, \hrep(T) > \cutoff}\hrep(T) \cdot\ind\left[\val \notin R_T\right]} 
= 
\expect{\val \leftarrow \dist}{\sum_{T \subseteq \items, \hrep(T) > \cutoff}{\hrep(T) \cdot \ind\left[\exists T', \hrep(T') > \hrep(T)\right] }} & &
\nonumber \\
= &
\sum_{T \subseteq \items}{\expect{\val \leftarrow \dist}{\hrep(T) \cdot \ind\left[\hrep(T) > \cutoff \wedge \val \notin R_T\right]}}  \text{ ~~~~~~~~~~~~~~~~~~(by linearity of expectation)} &\nonumber\\
= &
\sum_{T \subseteq \items}\sum_{x > \cutoff, f_T(x) > 0} x \cdot f_T(x) \cdot \Pr_{\dist_{-T}}\left[\exists T',  \hrep(T') > x\right]  \text{ ~~~~~~~~(by independence across hyperedges)} & \nonumber 
\end{align}

From here, we use essentially the same lemma from~\cite{cai2016duality}. We have replaced their $\SREV$ with $\BREV$, but the proof is identical. 
\begin{lemma}[\cite{cai2016duality}]\label{lem:cdwtail}
For all $x, T$, $x \cdot \Pr_{\hrep \gets \distH_{-T}}[\exists T', \hrep(T') > x] \leq \BREV$.
\end{lemma}

\begin{proposition}\label{prop:tail}
$\TAIL \leq \left(\sum_{T \subseteq \items} \Pr[\hrep(T) > \cutoff]\right) \cdot \BREV$.
\end{proposition}

\paragraph{Setting the Cutoff.}
Finally, we just need an appropriate choice of $\cutoff$. 
We'll choose to set $\cutoff$ such that $\sum_{T \subseteq \items} \Pr[\hrep(T) > \cutoff] = k$ for the appropriate choice of $k$. 
We first show how to relate $\cutoff$ to $\BREV$. Lemma~\ref{lem:prob} below is well-known, but we provide a proof in Appendix~\ref{app:proof:nonfav} for completeness.

\begin{lemma}\label{lem:prob}
Let $E_1,\ldots, E_k$ be independent events such that $\sum_{i} \Pr[E_i] = k$. Then $\Pr[\cup_i E_i] \geq 1-e^{-k}$.
\end{lemma}

\begin{corollary}\label{cor:brevcutoff}
If $\cutoff$ is such that $\sum_{T \subseteq \items} \Pr[\hrep(T) > \cutoff] = k$, then $\BREV \geq (1-e^{-k})\cutoff$.
\end{corollary}
\begin{proof}
	Apply Lemma~\ref{lem:prob} to the events $E_T = \{\hrep(T) > \cutoff\}$. Then the probability that there exists some hyperedge $T$ with $\hrep(T) > \cutoff$ is at least $(1-e^{-k})$. So the grand bundle will sell at price $\cutoff$ with probability at least $(1-e^{-k})$.
\end{proof}

We can now complete our bound for $\NONFAV$, and the proof of Theorem~\ref{thm:main}

\begin{proposition}\label{prop:nonfavorite}
$\NONFAV \leq 12\BREV$
\end{proposition}
\begin{proof}
Combine Propositions~\ref{prop:core} and~\ref{prop:tail} taking $\cutoff$ such that $\sum_T \Pr[\hrep(T) > \cutoff] =1.66$.
\end{proof}

\begin{proof}[Proof of Theorem~\ref{thm:main}]
	Simply combine Propositions~\ref{prop:single} and~\ref{prop:nonfavorite} with Corollary~\ref{cor:benchmark}.
\end{proof}

\section{Lower bounds}\label{sec:lowerBound}
The following proposition shows that the factor $d$ approximation (established in Theorem~\ref{thm:main}) is tight (up to constant factors), even when there are no substitutes ($\feas = 2^\items$).


\begin{proposition} \label{prop:lowerBound:weAreTight}
    There exists a distribution $\dist$ with complementarity $d$, for which \\
$\REV \geq \frac{d}{4}\max\{\BREV, \SREV \}$.
\end{proposition}

Furthermore, we argue that this parameter correctly characterizes the degree of complementarity in our setting.
Specifically, in Proposition~\ref{prop:lowerBound:othersAreBad}, we establish extremely high lower bounds (as a function of the complementarity degree) on the approximation ratio that can be obtained by $\max\{\BREV, \SREV \}$ for previous measures of complementarity from the literature. In what follows we give informal definitions of the different measures of complementarities and state their lower bounds. Formal definitions are deferred to Appendix~\ref{app:lowerBound}.


A valuation is in PH-$k$ \cite{abraham2012combinatorial} if its hypergraph representation $\hrep$ has only positive hyperedges $S$ of size at most $k$.
The supermodular  degree (SM) \cite{feige2013welfare}, roughly,  measures the distance of a valuation from being submodular; it ranges between $1$ to $m$.
A valuation is in PS-$k$ if in its hypergraph representation every item shares a positive hyperedge with at most $k$ other items.
It can be shown that PS-$k$ $\subseteq$ SM-$k$, thus every lower bound on PS-$k$
carries over to SM-$k$. Also, these lower bounds trivially hold for ``maximum over PH''~\cite{feige2015unifying} and ``maximum over PS''~\cite{feldman2016simple} hierarchies.
The following proposition asserts the lower
bounds for the aforementioned hierarchies. 

\begin{proposition} \label{prop:lowerBound:othersAreBad} The following hold for distributions in our settings, where hyperedges values $\hrep(T)$ are independently drawn, and $\val(S) = \sum_{T \subseteq S}\hrep(T)$.
    \begin{enumerate}
        \item \label{prop:lowerBound:PHk}
            There exists a distribution $\dist$ with only PH-$k$ valuations in the support, for which $\REV \geq \frac{1}{2m} \sum_{1 \leq i \leq k}\binom{m}{i} \max\{\BREV, \SREV \}$. E.g., for PH-$2$, $\REV \geq \Omega(m) \cdot \max\{\BREV, \SREV \}$.
        \item \label{prop:lowerBound:PSk}
            There exists a distribution $\dist$ with only PS-$k$ valuations in the support, for which $\REV \geq \frac{2^{k+1} - 1}{2(k+1)} \max\{\BREV, \SREV \}$.
    \end{enumerate}
\end{proposition}

Consider a set of hyperedges $E$ (to be defined per-case).
Index the hyperedges with integers in $\{1+a,2+a,  \ldots, \cardinality{E}+a\}$ (we abuse notation and use $e$ both for index and hyperedge, i.e., set of items).
The product distribution $\distH$ has $f_e(0) = 1$ for all $e\not \in E$,
and for every $e \in E$, set $f_e(0) = 1-2^{-e}$, and $f_e(2^e) = 2^{-e}$.
Let $\dist$ be the distribution that samples $\hrep \gets \distH$ and returns $\val(S) = \sum_{T \subseteq S} \hrep(T)$.

\begin{proposition}\label{prop:lb1}
	For the above distribution $\dist$, we have $\REV(\dist) \geq \cardinality{E}$, but $\SREV(\dist) \leq 2m$ and $\BREV(\dist) \leq 2$.
\end{proposition}
\begin{proof}
	First, consider the random variable $\val(\items)$. We have $\val(\items) \leq \sum_{e =1+a}^{\cardinality{E}+a} \hrep(e)$. For any price $p$, in order to have $\val(\items) \geq p$, we must have $\hrep(e) > 0$ for some $e \geq \log p$, as $\sum_{e=1+a}^{\log p-1} 2^e  = p - 2^{a+1}< p$.\footnote{$\sum_{e=1+a}^{n-1}{2^{e}} = 2^{n} - 2^{a+1} $}
	Note that the there is no reason to price below $2^{1+a}$.
	But also, by union bound, the probability that this occurs is at most
	$\sum_{e \geq \log p} 2^{-e}    \leq 2^{1-\log p} \leq 2/p$. So for any price $p$ we could set on the grand bundle, it sells with probability at most $2/p$, so $\BREV \leq 2$.
	
	Similarly, for any price $p_i$, in order for the buyer to possibly be willing to purchase item $i$, we must have $\sum_{e \ni i} \hrep(e) \geq p_i$.
	Again, in order for this to happen, we must have $\hrep(e) > 0$ for some $e \geq \log p_i$, $e \ni i$. And again by union bound, the probability that this occurs is at most $2/p_i$. So for any price $p_i$ we could set on item $i$, the probability that the buyer is possibly willing to purchase item $i$ is at most $2/p_i$, so $\SREV \leq 2m$.
	
	Consider however the following mechanism, which essentially sells the hyperedges in $E$ separately.
	The mechanism allows the buyer to purchase any set $S$ she chooses, and charges price $2^S$.
	By union bound,\footnote{$\sum_{e=1+a}^{n+a}{2^{-e}} =  2^{-a} - 2^{-n-a}$ therefore its complement is at least $1-2^{a}$} the probability that $\val \equiv 0$ is at least $1 - 2^{-a}$. Therefore, whenever $\hrep(e) > 0$, with probability at least $1 - 2^{-a}$, the buyer will choose to purchase exactly the set $e$ and pay $2^e$. So the revenue is at least $\sum_{e=1+a}^{\cardinality{E}+a} 2^{-e} \cdot 2^e \cdot (1 - 2^{-a}) = \cardinality{E}\cdot (1 - 2^{-a}) $. Taking $a \rightarrow \infty$ completes the proof.
\end{proof}
Let us now see how proposition~\ref{prop:lb1} implies proposition~\ref{prop:lowerBound:weAreTight} and proposition~\ref{prop:lowerBound:othersAreBad}.
\begin{proof}[proof of proposition~\ref{prop:lowerBound:weAreTight}]
	Consider a $d$ regular graph $(\items, E)$ over $m$ nodes. 
	By definition, every node is contained in exactly $d$ hyperedges.
	Therefore, if $E$ is the set of hyperedges used to construct $\dist$ prior to proposition~\ref{prop:lb1}, then $\dist$ has complementarity $d$, and $\cardinality{E} = md/2$. 
\end{proof}
\begin{proof}[proof of proposition~\ref{prop:lowerBound:othersAreBad}]
	To show \ref{prop:lowerBound:othersAreBad}.\ref{prop:lowerBound:PHk}, consider the set of all hyperedges of size at most $k$, and apply proposition~\ref{prop:lb1}.
	To show proposition~\ref{prop:lowerBound:othersAreBad}.\ref{prop:lowerBound:PSk}, assume for simplicity that $m$ is divisible by $k+1$. 
	Partition $\items$ to $m/(k+1)$ sets  $\items_1, \items_2, \ldots \items_{m/(k+1)}$, all of size $k+1$, and consider all hyperedges $S \subseteq \items_i$ for all $i$.
	Every item $i$ in $\items_j$ has neighbors only from $\items_j$, therefore every valuation in the support is from PS-$k$.
	The number of hyperedges is $\frac{m}{k+1} \cdot( 2^{k+1} - 1)$. Applying proposition~\ref{prop:lb1} completes the proof.
\end{proof}

\appendix
\section{Background on Duality Framework} \label{app:duality}
We first recall the duality approach of~\cite{cai2016duality}:

\begin{definition}\label{def:duality}[Reworded from~\cite{cai2016duality}, Definitions~2 and~3]
	A mapping $\lambda: V \times V \rightarrow \reals^+$ is \emph{flow-conserving} if for all $\val \in V$: $\sum_{\val '\in V} \lambda(v, v') \leq f(v)+ \sum_{\val ' \in V} \lambda(v', v)$.\footnote{This is equivalent to stating that there exists a $\lambda(v, \bot) \geq 0$ such that $\lambda(v, \bot)+\sum_{\val '\in V} \lambda(v, v') = f(v)+ \sum_{\val ' \in V} \lambda(v', v)$, which might look more similar to the wording of Definition~2 in~\cite{cai2016duality}.} The \emph{virtual transformation} associated with $\lambda$, $\Phi^\lambda$, is a transformation from valuation functions in $V$ to valuation functions in $V^\times$ (the closure of $V$ under linear combinations) and satisfies:\footnote{That is, $\Phi^\lambda(\val)$ is a (possibly negative) function from $2^\items$ to $\reals$, and satisfies $\Phi^\lambda(\val)(S) = \val(S) - \frac{1}{f(\val)}\sum_{\val ' \in V} \lambda(\val ', \val )(\val ' (S) - \val (S) ).$ for all $S \subseteq \items$.}
	
	$$\Phi^\lambda(\val)(\cdot) = \val(\cdot) - \frac{1}{f(\val)}\sum_{\val ' \in V} \lambda(\val ', \val )(\val ' (\cdot) - \val (\cdot) ).$$
\end{definition}

In the above definition, one should interpret $\lambda(\cdot, \cdot)$ as being potential Lagrangian multipliers for incentive constraints in a certain LP to find the revenue-optimal mechanism, and think of $f(\val)$ flow going into each $\val$ from some super source, $\lambda(\val, \val ')$ flow going from $\val$ to $\val '$, and all excess flow (that enters $\val$ but doesn't leave) as going from $\val$ to a super sink. Note that whether or not a given $\lambda$ is flow-conserving depends on the population $\dist$. Cai et al. show that Lagrangian multipliers that satisfy the above flow conservation constraint yield upper bounds of the following form.

\begin{theorem}\label{thm:duality}[Reworded from~\cite{cai2016duality}, Theorem~10]
	Let $\M$ be any truthful mechanism where a bidder with type $\val$ receives items $X(\val)$ and pays $p(\val)$. Then for all flow-conserving $\lambda$, the expected revenue of $\M$ is upper bounded by its expected virtual welfare with respect to $\lambda$. That is:
	
	$$\expect{\val \leftarrow \dist}{p(\val)} \leq \expect{\val \leftarrow \dist}{\Phi^\lambda(\val)(X(\val))}.$$
\end{theorem}

As an immediate corollary, we can obtain the following upper bound on the revenue of any truthful mechanism by observing that the bound in Theorem~\ref{thm:duality} is maximized when $X(\val)$ is deterministically $\argmax_{S \subseteq \items}\{\Phi^\lambda(\val)(S)\}$.

\begin{corollary}\label{cor:duality}
	For all $\dist$, and all flow-conserving $\lambda$, we have:
	
	$$\REV (\dist) \leq \expect{\val \leftarrow \dist}{\max_{S \subseteq 2^M} \Phi^\lambda(\val)(S)}.$$
	
\end{corollary}

We begin this section by defining our flow-conserving $\lambda$ and the resulting $\Phi^\lambda$. Readers familiar with~\cite{cai2016duality} will recognize it as the natural generalization of their flow to our setting, and we will make the language as similar as possible.

We will break $V$ into $2^m-1$ different regions, depending on which hyperedge is the most valuable to a buyer with value $\val$.
Specifically, we say that $\val$ is in region $R_A$ if $A = \argmax_{T \subseteq \items} \{\hrep(T) \}$, with ties broken lexicographically.
{Recall that $\dist$ is established by drawing $\hrep$ from the product distribution $\distH$ and the returned valuation $\val$ satisfies $\val(S) = \max_{T \subseteq S, T \in \feas}\{\sum_{U \subseteq T} \hrep(U)\}$. }
Then consider the following flow:

\begin{definition}[Flow for our benchmark]\label{def:ourflow}
	If $\val \in R_A$, define $\hrep ' (T) = \hrep(T)$, for all $T \neq A$, and define $\hrep '(A) = \min_{x > \hrep(A)}\{x : f_A(x) > 0\}$. Set $\lambda(v', v) = \Pr_{x \gets \distH_A}[x \geq \hrep'(A)] \cdot \prod_{T \neq A} f_T(\hrep'(T)) = f(\val) \cdot \frac{\Pr_{x \gets \distH_A}[x \geq \hrep'(A)]}{f_A(w(A))}$ for the $v'(\cdot)$ such that $v'(S) = \max_{T \subseteq S, T \in \feas}\{\sum_{U \subseteq T} \hrep '(U)\}$ for all $S$, and $\lambda(v'', v) = 0$ for all other $v''$.
\end{definition}

\begin{proposition}\label{prop:flow}
	The $\lambda(\cdot, \cdot)$ from Definition~\ref{def:ourflow} is flow-conserving. Moreover, if $\val(\cdot)$ is such that $\val(S) = \max_{T \subseteq S, T \in \feas}\{\sum_{U \subseteq T}\hrep(U)\}$, and $\val \in R_A$, then $\Phi^\lambda$ satisfies the following:
	\begin{align*}\Phi^\lambda(\val)(S) \leq \max_{T \subseteq S, T \in \feas}\{\sum_{U \subseteq T, U \neq A} \hrep(U)\} + \max\{0,\virtVal_A(\hrep(A))\} \leq \max_{T \in \feas} \{\sum_{U \subseteq T, U \neq A} \hrep(U)\} + \max \{0,\virtVal_A(\hrep(A))\}.
	\end{align*}
\end{proposition}
\begin{proof}
	That $\lambda(\cdot, \cdot)$ is flow-conserving is clear: every $\val \in R_A$ has total incoming flow of $f(\val) \cdot \frac{\Pr_{x \gets \distH_A}[x \geq \hrep(A)]}{f_A(\hrep(A))}$ ($f(\val)$ of this comes from the source, the remaining $f(\val) \cdot \frac{\Pr_{x \gets \distH_A}[x > \hrep(A)]}{f_A(\hrep(A))}$ comes from other types in $R_A$). Every $\val \in R_A$ also has outgoing flow either equal to $0$ (if decreasing the value of $\hrep(A)$ moves the resulting $\val '$ out of $R_A$), or exactly $f(\val) \cdot \frac{\Pr_{x \gets \distH_A}[x \geq \hrep(A)]}{f_A(\hrep(A))}$ (otherwise). In either case, the flow out is at most the flow in.
	
	Let's now compute $\Phi^\lambda(\val)(S)$. Plugging into Definition~\ref{def:duality}, we get:
	
	$$\Phi^\lambda(\val)(S) = \val(S) - \frac{(\val'(S) - \val(S))\Pr_{x \gets \distH_A}[x \geq \hrep(A)]}{f_A(\hrep(A))}.$$
	
	Recall that $\val '(S) \geq \val (S)$ for all $S$, and therefore $\Phi^\lambda(\val)(S) \leq \val(S)$ for all $S$. Now there are two cases to consider: In the first case, maybe $\max_{T \subseteq S, T \in \feas} \{\sum_{U \subseteq T, U \neq A} \hrep(U)\} = \val (S)$. In other words, the set in $\feas$ ``chosen'' by a consumer with valuation $\val$ doesn't contain $A$. In this case, we immediately get that $\Phi^\lambda(\val)(S) \leq \val(S) = \max_{T \subseteq S, T \in \feas}\{\sum_{U \subseteq T, U \neq A} \hrep(U)\}$, as desired.
	
	In the second case, maybe $\max_{T \subseteq S, T \in \feas} \{\sum_{U \subseteq T, U \neq A} \hrep(U)\} < \val (S)$. In other words, the set in $\feas$ ``chosen'' by a consumer with valuation $\val$ contains $A$. In this case, increasing $\hrep(A)$ by any $x > 0$ increases $\val (S)$ by exactly $x$. Therefore, we have $\val '(S) = \val (S) + \hrep ' (A) - \hrep (A)$, and therefore:
	
	\begin{align*}
	\Phi^\lambda(\val)(S) &= v(S) - \frac{(\hrep '(A) - \hrep(A)) \Pr_{x \gets \distH_A}[x \geq \hrep(A)]}{f_A(\hrep(A))}\\
	&= \max_{T \subseteq S, T \in \feas} \{\sum_{U \subseteq T} \hrep(U)\} - \frac{(\hrep '(A) - \hrep(A)) \Pr_{x \gets \distH_A}[x \geq \hrep(A)]}{f_A(\hrep(A))}\\
	&\leq \max_{T \subseteq S, T \in \feas} \{\sum_{U \subseteq T, U \neq A} \hrep(U)\} + \hrep(A) - \frac{(\hrep '(A) - \hrep(A)) \Pr_{x \gets \distH_A}[x \geq \hrep(A)]}{f_A(\hrep(A))}\\
	&= \max_{T \subseteq S, T \in \feas} \{\sum_{U \subseteq T, U \neq A} \hrep(U)\} + \virtVal_A(\hrep(A)).
	\end{align*}
	
	The last line uses the definition $\virtVal_A(\hrep(A)) = \hrep(A) - \frac{(\hrep '(A) - \hrep(A)) \Pr_{x \gets \distH_A}[x \geq \hrep(A)]}{f_A(\hrep(A))}$, which may seem unfamiliar to readers more familiar with virtual values for continuous distributions. Indeed, this is the right generalization of Myerson's $\virtVal(\cdot)$ for continuous distributions to the discrete setting, and we refer the interested reader to~Section~4 of~\cite{cai2016duality} for more discussion.
\end{proof}

\paragraph{Ironing.} The astute reader will notice that when $\distH_S$ is irregular, the bound we probably want above would replace $\virtVal_A(\cdot)$ with $\bar{\virtVal}_A(\cdot)$.~\cite{cai2016duality} shows how to design a flow that accomplishes this essentially by adding cycles to $\lambda$ between adjacent types to ``iron out'' any non-monotonicities, but for their setting of additive buyers. The exact same approach will work here. We omit a proof and refer the reader to~\cite{cai2016duality} for more detail. This allows us to prove corollary~\ref{cor:benchmark}.
\begin{proof}[proof of corollary~\ref{cor:benchmark}]
	Simply combine Corollary~\ref{cor:duality} and Proposition~\ref{prop:flow}, after replacing $\virtVal(\cdot)$ in Proposition~\ref{prop:flow} with $\bar{\virtVal}(\cdot)$.
\end{proof}

\section{Missing Proofs From Section~\ref{sec:prelim} and \ref{subsec:virtValAndBenchmark}}  \label{app:proofs:prelim}
\begin{definition} \label{def:FOSD}
	A Random variable $X$ is first-order stochastically dominated (FOSD) by random variable $Y$ if for every $x$, $\Pr\left[X \geq x\right]\leq \Pr\left[Y \geq x\right]$.
	\\
	{\bf Note: } If $X$ is FOSD by $Y$ then $\expect{}{X} \leq \expect{}{Y}$.
\end{definition}

\begin{proof}[proof of theorem~\ref{thm:pricing}]
	Let $\ind_q$ be an independent indicator random variable that equals $1$ with probability $q$.
    Let $\{X_S\}_S$ be non-negative independent random variables that are drawn from the independent distributions.
    
    Consider a tie breaking rule among the sets, and let the event $X_S = \max_T\{X_T\}$ be true only when $S$ also wins in the tie breaking rule. Set $q_S = \Pr\left[X_S = \max_T\{X_T\}\right]$. So $\sum_S q_S = 1$.
    set $t_S$ s.t. $\Pr\left[X_S \geq t_S \right] = q\cdot q_S$. 

    Let us see that the random variable   $\ind_q \cdot X_S \cdot \ind\left[ X_S = \max\{X_T\} \right]$ is FOSD (definition~\ref{def:FOSD}) by 
    $X_S \cdot \ind\left[ X_S \geq t_S \right]$.
    
    For every $x \geq t_S$, it holds that $\Pr\left[\ind_q \cdot X_S \cdot \ind\left[ X_S = \max\{X_T\} \right] \geq x \right] \leq q \Pr\left[X_S \geq x \right]$, while $\Pr\left[ X_S \cdot \ind \left[X_S \geq  t_S\right] \geq x \right] = \Pr\left[X_S \geq x\right]$.
    
    For every $x < t_S$, it holds that $\Pr\left[\ind_q \cdot X_S \cdot \ind\left[ X_S = \max\{X_T\} \right] \geq x \right] \leq q\cdot q_S$ by definition of $q_S$, while $\Pr\left[ X_S \cdot \ind \left[X_S \geq  t_S\right] \geq x \right] = \Pr\left[X_S \geq  t_S\right] = q\cdot q_S$ by definition of $t_S$. We get that:

    \begin{align*}
    \expect{}{\max_S \{X_S\}} 
    & = 
    \sum_S \expect{}{X_S \cdot \ind\left[ X_S = \max\{X_T\} \right] }   \\
    & = 
    \frac{1}{q} \sum_S \expect{}{\ind_q \cdot X_S \cdot \ind\left[X_S = \max\{X_T\} \right] }   \\
	& \leq 
	\frac{1}{q}\sum_S \expect{}{X_S \cdot \ind\left[ X_S \geq t_S \right] } 
    \end{align*}
   Let $X_S$ be the random variable that first draws $x \gets \distH_S$ and returns $\max\{0, \virtVal_S(x)\}$.
    	Assume the distributions are regular, and refer to  \cite{chawla2010multi} for the irregular case. 
    	As $t_S \geq 0$ we get:
    	\begin{align*}
    		       \expect{x\gets \distH_S}{\max_S \{\virtVal_S(x), 0\}}  \leq  \sum_S \expect{x\gets \distH_S}{\virtVal_S(x)\cdot \ind\left[\virtVal_S(x) \geq t_S \right] }
    	\end{align*}
    	Observe that the above term for each $S$ is the expected virtual value of the mechanism that allocates to a bidder with value $x$ if $x$ exceeds $p_S = \inf \{x: \virtVal_S(x) = t_S\}$. This allocation is achieved by posting a price $p_S$.
    	By Myerson's payment identity:
    	\begin{align*}
    		\expect{x\gets \distH_S}{\virtVal_S(x)\cdot \ind\left[\virtVal_S(x) \geq t'_S \right] } =
    		\expect{x\gets \distH_S}{p_S \cdot \Pr\left[x \geq p_S \right] }
    	\end{align*}
    	This concludes property 1. Property 2 follows by monotonicity of $\virtVal_S$ (regularity of $\distH_S$, for the irregular case refer to \cite{chawla2010multi}):
    	\begin{align*}
    		\sum_S \Pr_{\distH_S}\left[x \geq p_S \right] = \sum_S \Pr_{\distH_S}\left[\virtVal_S(x) \geq t_S \right]  = \sum_S q \cdot q_S = q
    	\end{align*}
    	
%
%
%
    	    


\end{proof}

\begin{proof}[Proof of Theorem~\ref{thm:main}]
    Simply combine Propositions~\ref{prop:single} and~\ref{prop:nonfavorite} with Corollary~\ref{cor:benchmark}, to get:
    \begin{align*}
    \REV(\dist) \leq 4d \SREV^*(\dist) + 16 \BREV(\dist) \leq
    (4d + 16)\max\{\SREV^* , \BREV\}
    \end{align*}
\end{proof}


\section{Missing Proofs From Section~\ref{sec:non-favorite}} \label{app:proof:nonfav}
\begin{proof}[proof of lemma~\ref{lem:coretail}]
    The proof follows from the following algebra:
    \begin{align*}
    \mbox{(NON-FAVORITE)}
    = &\expect{\val \leftarrow \dist}{\max_{S \in \feas}\{\sum_{T \subseteq S} \hrep(T)\cdot\ind\left[\val \notin R_T \right]\}} \\
    = &\expect{\val \leftarrow \dist}{\max_{S \in \feas}\{\sum_{T \subseteq S} \hrep(T)\cdot\ind \left[\hrep(T) \leq t\right]\cdot \ind\left[\val \notin R_T \right]+ \hrep(T)\cdot \ind\left[ \hrep(T) > \cutoff \right]\cdot\ind\left[\val \notin R_T \right]\}} \\
    \leq &\expect{\val \leftarrow \dist}{\max_{S \in \feas}\{\sum_{T \subseteq S} \hrep(T)\cdot\ind \left[\hrep(T) \leq \cutoff \right]\}} \\
    &~~~~~~~~+\expect{\val \leftarrow \dist}{\max_{S \in \feas}\{\sum_{T \subseteq S} \hrep(T)\cdot \ind\left[ \hrep(T) > \cutoff \right]\cdot\ind\left[\val \notin R_T \right]\}}\\
    \leq &\expect{\val \leftarrow \dist}{\max_{S \in \feas}\{\sum_{T \subseteq S} \hrep(T)\cdot \ind\left[\hrep(T) \leq \cutoff \right]\}} \\
    &~~~~~~~~+\expect{\val \leftarrow \dist}{\sum_{T |\hrep(T) > \cutoff} \hrep(T)\cdot\ind\left[\val \notin R_T \right]\}}
    \end{align*}
\end{proof}

\begin{proof}[proof of corollary~\ref{cor:schechtman}]
    By corollary~12 in \cite{schechtman2003concentration}, we know that for all $k >0$:
    \begin{align} \label{eq:schechtman}
    \Pr\left[\val(N) \geq 3 \cdot a  + k\cdot c \right] \leq \min \{1,4 \cdot 2^{-k}\}
    \end{align}
    Substituting $x = 3 \cdot a  + k\cdot c $ gets $k = (x-3a)/c$. Therefore, equation~(\ref{eq:schechtman}) becomes meaningful only when $4 \cdot 2^{-k} \leq 1$, i.e., when $ x  \geq 2 c + 3 a$.
    Computing the expected value of $\val(N)$ gives:
    \begin{align*}
    \int_0^{\infty}{\Pr\left[\val(N) > x\right] dx}
    \leq
    \int_0^{\infty}{\min \{1,4 \cdot 2^{(3a-x)/c}\}  dx}
    =
    2 c + 3a+
    4 \cdot 2^{3a/c} \cdot \int_{2\cdot c + 3a}^{\infty}{ 2^{ -x/c}\cdot dx}
    \end{align*}
    Computing the integral gives:
    $
    -\frac{c}{\ln 2}\left[ 2^{-x/c}\right]_{2c + 3a}^{\infty} =
    \frac{c}{\ln 2}\cdot 2^{-\tfrac{2c + 3a}{c}} =
    \frac{c}{4 \ln 2}\cdot 2^{-3a/c}
    $,
    which, plugged back to the equation concludes that
    $$
    \expect{}{\val(N)} \leq 2c + 3a +  \frac{c}{\ln 2}
    $$
    as desired.
\end{proof}

\begin{proof}[proof of lemma~\ref{lem:cdwtail}]
    For any $x$, we can set price $x$ on the grand bundle. It will sell with probability at least $\Pr_{\hrep \gets \distH_{-T}}[\exists T', \hrep(T') > x]$, as whenever there is a single hyperedge with contribution $x$, certainly the buyer's value for the grand bundle is at least $x$. Therefore, $\BREV \geq x \cdot \Pr_{\hrep \gets \distH_{-T}}[\exists T', \hrep(T') > x]$.
\end{proof}

\begin{proof}[proof of proposition~\ref{prop:tail}]
    By Lemma~\ref{lem:cdwtail}, We get:
    $$\sum_{T \subseteq \items}\sum_{x > \cutoff, f_T(x) > 0} x \cdot f_T(x) \cdot \Pr_{\hrep \gets \distH_{-T}}\left[\exists T',  \hrep(T') > x\right] \leq \sum_{T \subseteq \items} \sum_{x > t, f_T(x) > 0} f_T(x) \cdot \BREV = \sum_{T \subseteq \items} \Pr[\hrep(T) > \cutoff] \cdot \BREV.$$
\end{proof}

\begin{proof}[proof of lemma \ref{lem:prob}]
	By independence:
	\begin{align*}
	\Pr\left[\cup_i E_i \right] = 1 - \prod_i \left(1- \Pr\left[E_i\right]\right)
	\end{align*}
	So if we define $q_i = \Pr[E_i]$, we want to maximize $\prod_i \left(1 -q_i\right)$ subject to $\sum_i q_i = k$.
	Using a Lagrangian multiplier of $\lambda$ on the constraint $\sum_i q_i = k$, we get a new objective of:
	\begin{align*}
	\prod_i \left(1 -q_i \right) +\lambda \cdot (\sum_i{q_i} ) - \lambda k
	\end{align*}
	We see that the partial with respect to $q_i$ of the above is exactly $-\prod_{j \neq i} (1-q_j) + \lambda$. So setting $q_i = k/n$ for all $i$, and $\lambda = (1-k/n)^{n-1}$, we see that $\sum_i q_i = k$ and the partial of the Lagrangian with respect to $q_i$ is $0$ for all $i$. Therefore, this is the optimal solution. At $q_i = k/n$ for all $i$, we have $\prod_i (1-k/n) = (1-k/n)^n \leq e^{-k}$.
\end{proof}

\section{Lower Bounds} \label{app:lowerBound}
Below We formally define the previously mentioned complementarity measures.
\begin{definition}\cite{abraham2012combinatorial,feige2015unifying}(MPH)
	A valuation $\val$ is positive hypergraph (PH) of degree at most $k$ if there exists a hyperedge weight function $\hrep \geq 0$, $\hrep(T)  = 0$ for all $\cardinality{T} > k$, so that $\val(S) = \sum_{T\subseteq S}\hrep(T)$.
	
	A valuation $\val$ is maximum over PH (MPH) of degree at most $k$ if there exists a collection $L$ of such hyperedge weight functions, so that $\val(S) = \max_{\ell \in L}\{\sum_{T\subseteq S}\hrep_\ell(T) \}$.
	\end{definition}
	\begin{definition}\cite{feldman2016simple}(MPS)
		A valuation $\val$ is positive supermodular (PS) of degree at most $k$ if there exists a hyperedge weight function $\hrep \geq 0$, so that for every item $i$, it holds that \\
		$\cardinality{\{j \in \items : \exists T, \hrep(T) >0,  \{j, i\} \subseteq T \} } \leq k$, i.e., item $i$ has at most $k$ neighbors (other items that share a positive hyperedge with it).
		
		A valuation $\val$ is maximum over PS (MPS) of degree at most $k$ if there exists a collection $L$ of such hyperedge weight functions, so that $\val(S) = \max_{\ell \in L}\{\sum_{T\subseteq S}\hrep_\ell(T) \}$.
		\end{definition}
		\begin{definition}\cite{feige2013welfare}(SM)
			A valuation $\val$ is supermodular (SM) of degree at most $k$ if for each item $i$, the number of items $i'$ so that
			there exists a set $S_{i'} \not \ni i$ so that $\val(S_{i'} \cup i) - \val(S_{i'}) > \val(S_{i'}\setminus \{i'\} \cup \{i\}) - \val(S_{i'} \setminus \{i'\})$ is at most $k$, i.e., $i$'s marginal contribution to a set may {\em increase} by adding another item, to at most $k$ different items.
			\end{definition}


\bibliographystyle{abbrv}
\bibliography{AGT}

\end{document}